\documentclass[aps,prl,reprint,superscriptaddress]{revtex4-2}

\usepackage{amsmath,amssymb,mathtools,amsthm}
\usepackage{graphicx}

\newcommand{\ii}{\mathrm i}
\newcommand{\dd}{\mathrm d}
\newcommand{\Tr}{\operatorname{Tr}}
\newcommand{\Var}{\operatorname{Var}}

\newcommand{\Ex}{\mathbb E}
\newcommand{\norm}[1]{\lVert#1\rVert}
\newcommand{\ket}[1]{\lvert#1\rangle}
\newcommand{\bra}[1]{\langle#1\rvert}

\newtheorem{theorem}{Theorem}
\newtheorem{proposition}[theorem]{Proposition}

\theoremstyle{remark}

\makeatletter
\@ifundefined{frontmatter@RRAPformat}{}{%
  \renewcommand*\frontmatter@RRAPformat[1]{}}
\makeatother

\begin{document}
\title{Laughlin quasihole geometry from a single snapshot ensemble}

\author{Kaushlendra Kumar}
\affiliation{School of Mathematical Sciences, Queen Mary University of London, Mile End Road, London E1 4NS, United Kingdom\\
\textup{kaushlendra.kumar@qmul.ac.uk}}

\date{}

\begin{abstract}
Can a probability distribution measured in one basis determine complex quantum geometry? The answer is affirmative for a lattice Laughlin quasihole. The exact occupation law at one generic reference position, together with the known analytic quasihole factor, fixes the complete complex Gram kernel.  A finite snapshot ensemble estimates this kernel without preparing another member of the family, and thereby determines finite-distance overlaps, Bargmann phases, the quantum metric, and the Berry curvature. The same analytic structure confines the family to an exact projective subspace whose dimension grows at most linearly with particle number. Exact enumeration of a Nielsen--Cirac--Sierra state demonstrates finite-shot reconstruction of both the metric and a geometric phase. Moreover, a R\'enyi-2 divergence sets the statistical range of the reconstruction while, at nondegenerate points, occupation readout also attains the local single-copy multiparameter information bound.
\end{abstract}

\maketitle

The quantum metric and Berry curvature are relational quantities.  They describe, respectively, how a ray changes as a parameter moves and the phase accumulated around a loop.  Standard routes therefore compare nearby preparations, interfere distinct states, or drive the system and measure its response \cite{ProvostVallee1980,KolodrubetzEtAl2017}.  An occupation snapshot contains less obvious information: it is one projective measurement in one basis, recording only where the particles were found.

Quantum-gas microscopes now resolve occupation patterns shot by shot, and a bosonic lattice Laughlin state has been prepared with ultracold atoms \cite{LeonardEtAl2023}.  Recent work has also shown that one bulk snapshot data set can reveal observables of defects that were never physically introduced \cite{SarmaEtAl2026}.  These developments make a sharper question experimentally meaningful.  Can measurements at one quasihole position determine the \emph{complex} relations among states that were never prepared?

For a Laughlin quasihole, the answer is yes.  Moving the quasihole changes each occupation amplitude by a \emph{known complex multiplier}.  Averaging products of these multipliers against the measured reference distribution cancels the unknown reference phases and yields the full overlap kernel.  The local quantum geometric tensor (QGT) is recovered as its coincident-point derivative.

Place \(N_s\) lattice sites at arbitrary points \(z_i\in\mathbb C\) and consider \(N_p\) hard-core particles.  A quasihole of integer strength \(p\) is localized at the externally controlled parameter \(w\in\mathbb C\).  The coordinate \(w\) labels the state and is not a dynamical particle coordinate.  For the lattice Laughlin states \cite{NielsenCiracSierra2012,GlasserEtAl2016}, the amplitude of an occupation configuration \(\mathbf{n}=(n_1,\ldots,n_{N_s})\), \(n_i\in\{0,1\}\), \(\sum_i n_i=N_p\), factorizes as
\begin{align}
 f_\mathbf{n}(w)&:=\prod_i(w-z_i)^{p n_i},\qquad
 \Phi_\mathbf{n}(w)=\Phi_0(\mathbf{n})f_\mathbf{n}(w),\nonumber\\
 X_\mathbf{n}(w)&:=\partial_w\log f_\mathbf{n}(w)
 =p\sum_i\frac{n_i}{w-z_i}.
 \label{eq:family}
\end{align}
All Laughlin, lattice, and gauge factors that do not move with the quasihole are collected in \(\Phi_0\). No translation symmetry or regular lattice is needed.  The factorization places no constraint on the site geometry, and in fact lattice Laughlin states have been studied on irregular and fractal supports \cite{MannaEtAl2020}. We also use occupation-basis expansion of the quasihole state \(\ket{\Phi(w)}=\sum_\mathbf{n}\Phi_\mathbf{n}(w)\ket{\mathbf{n}}\), with normalisation \(Z(w)=\langle\Phi(w)|\Phi(w)\rangle\), and normalised state \(\ket{\psi(w)}=Z(w)^{-1/2}\ket{\Phi(w)}\) in the following.

The protocol requires only the reference state at a generic \(w_0\). Repeated occupation measurements sample the probability distribution $P_0({\bf n})$ and corresponding expectations for any observable $A$,
\begin{equation}
 P_0(\mathbf{n})=\frac{|\Phi_\mathbf{n}(w_0)|^2}{Z(w_0)},
 \qquad
 \Ex_0[A]=\sum_\mathbf{n}P_0(\mathbf{n})A(\mathbf{n}).
 \label{eq:reference}
\end{equation}
For every recorded configuration and every target position \(w\), the ratio \(r_\mathbf{n}(w)=f_\mathbf{n}(w)/f_\mathbf{n}(w_0)\) is known.  A single expectation therefore gives
\begin{align}
 \mathcal K(w,v)&:=\Ex_0\!\left[
 \overline{r_\mathbf{n}(w)}r_\mathbf{n}(v)\right],
 &\mathcal K(w,w)&=\frac{Z(w)}{Z(w_0)},\nonumber\\
 \langle\psi(w)|\psi(v)\rangle&=
 \frac{\mathcal K(w,v)}
 {\sqrt{\mathcal K(w,w)\mathcal K(v,v)}}.
 \label{eq:kernel}
\end{align}
Equation~\eqref{eq:kernel} is the central result of this Letter. The measured distribution fixes the statistical weights, while the known amplitude ratios supply their relative complex phases.  Their average therefore yields overlaps between unprepared states without a coherent comparison of separately prepared family members.  Unlike ordinary correlated sampling, which reweights diagonal expectation values \cite{FoulkesEtAl2001}, \eqref{eq:kernel} retains the complex bra--ket ratio and reconstructs off-diagonal overlaps and Bargmann phases.  Coherent cycle tests can obtain the same relational data from several prepared states \cite{OszmaniecBrodGalvao2024}.  Here the analytic quasihole structure provides a separate route for such estimations.

Configuration-dependent imperfections that are independent of \(w\) are absorbed at the reference point: their moduli enter \(P_0\), while their phases cancel.  Only violations of the known amplitude ratio change the reconstructed geometry.  The Supplemental Material gives an explicit norm bound on the resulting kernel error \cite{Supplemental}.  Readout errors are different in kind as they act on the recorded configuration rather than on the amplitude, so they replace \(P_0\) by a stochastic image and are not absorbed. An independently calibrated readout channel can be inverted before the average is taken.  Two reference positions give a check that needs no such calibration, because their reconstructions must agree wherever both remain statistically accessible.

\begin{figure*}[t!]
 \centering
 \includegraphics[width=\textwidth]{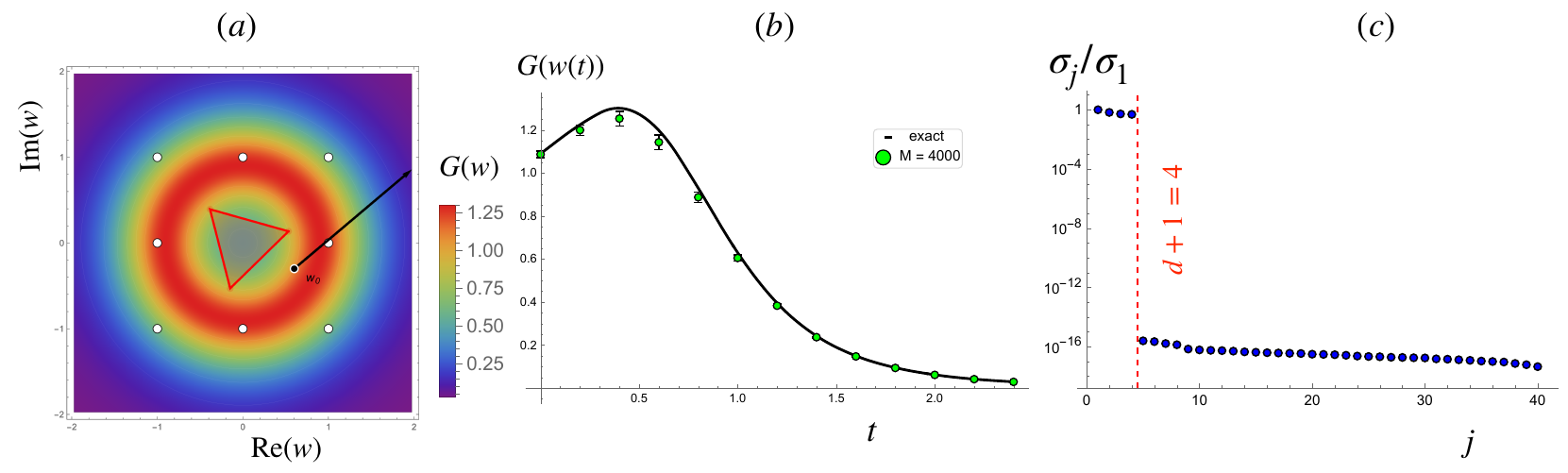}
 \caption{Single-reference reconstruction of quasihole geometry for an exactly enumerated \(q=2\) Nielsen--Cirac--Sierra state with \(N_s=8\), \(N_p=3\), and \(p=1\). (a) Exact scalar metric coefficient \(G(w)\).  White circles mark the lattice sites, the black point marks the prepared reference \(w_0\), the black ray is the reconstruction path, and the red triangle marks the three positions used for the Bargmann-phase test. (b) Reconstruction of \(G(w(t))\) along the ray with green points showing one \(M=4\times10^3\) realization, error bars giving one standard deviation over \(160\) independently generated reference ensembles while the black curve is exact. (c) Normalized singular values of the state matrix at \(40\) quasihole positions, showing exactly four nonzero values.}
 \label{fig:reconstruction}
\end{figure*}

The kernel also determines finite-separation geometry.  Products of reconstructed overlaps around a polygon are gauge-invariant Bargmann invariants \cite{Bargmann1964}.  Their phases approach the Berry holonomy as the polygon resolves a loop \cite{AvdoshkinPopov2023}.

The same kernel fixes a distance that is geometric rather than merely statistical. For the density matrix \(\rho_w=\ket{\psi(w)}\!\bra{\psi(w)}\), a finite spectral triple assigns to its state space the Connes distance, obtained as a supremum of state differences over the Lipschitz ball of a Dirac operator \cite{Connes1994}. For the scalar-anchored triple introduced in Appendix~C of \cite{Kumar2026anchor}, extended here to the many-body matrix sector, that supremum has a closed form,
\begin{equation}
 d_\Lambda(\rho_w,\rho_v)=\frac{2}{\Lambda}
 \sqrt{1-\frac{|\mathcal K(w,v)|^2}
 {\mathcal K(w,w)\mathcal K(v,v)}}.
 \label{eq:chord}
\end{equation}
The resulting consequences are threefold. The distance is derived from a Dirac operator rather than postulated, so the quasihole family carries a spectral geometry and not only a fidelity.  The scalar \(\Lambda>0\) is the only free parameter of that geometry, so one calibrated reference distance fixes every other distance in the family.  That geometry is in turn tied to a calibration-independent operational quantity, since \(P_{\rm succ}^{\rm opt}=\tfrac12+\tfrac{\Lambda}{4}d_\Lambda\) \cite{Helstrom1976} returns the equal-prior Helstrom probability, in which \(\Lambda\) cancels.  The scalar therefore sets the spectral units, and the same relation converts them into distinguishability. Finally, Connes distances on finite geometries are ordinarily computed from a known algebra.  Equation~\eqref{eq:chord} instead makes one accessible to measurement, since the right-hand side is built from reconstructed overlaps. The distance calibrates distinguishability but is not itself a probability.

The same analytic structure strongly constrains the family's projective span.  Since \(f_\mathbf{n}(w)\) has degree \(d=pN_p\),
\begin{equation}
 \ket{\Phi(w)}=\sum_{k=0}^{d}w^k\ket{\chi_k},
 \qquad
 \ket{\chi_k}=\sum_\mathbf{n}\Phi_0(\mathbf{n})a_k(\mathbf{n})\ket{\mathbf{n}}.
 \label{eq:compression}
\end{equation}
The positive moment matrix \(H_{kl}=\langle\chi_k|\chi_l\rangle\) stores the kernel: \(\mathcal K(w,v)=Z(w_0)^{-1}m_d(w)^\dagger Hm_d(v)\), where \(m_d(w)=(1,w,\ldots,w^d)^\top\).  If \(\operatorname{rank}H=r+1\), the minimal projective ambient space is \(\mathbb{CP}^r\), with \(r\le pN_p\).  Provided \(r\ge1\), the moving quasihole traces a one-complex-dimensional polynomial curve of degree at most \(d\).  Its projective span is independent of the ambient many-body dimension.  Rank one instead gives a constant projective ray.  In Fig.~\ref{fig:reconstruction}(c), the 56-component many-body state curve has projective span \(\mathbb{CP}^3\).

The local geometry follows by bringing the two kernel arguments together.  From \(\mathcal K(w,w)=Z(w)/Z(w_0)\),
\begin{align}
 G(w)&:=\partial_w\partial_{\bar w}\log\mathcal K(w,w)
 =\Var_w X,\nonumber\\
 \mathcal Q_{ab}&:=\langle D_a\psi|D_b\psi\rangle
 =G\begin{pmatrix}1&\ii\\-\ii&1\end{pmatrix}_{ab},
 \label{eq:QGT}
\end{align}
where \(\ket{D_a\psi}:=(1-\rho_w)\partial_a\ket\psi\).  It follows that the Fubini--Study metric on the \emph{quasihole-position parameter space} is \(g^{\rm FS}_{ab}=G\delta_{ab}\), while our Berry-curvature convention \cite{Berry1984} gives \(\Omega_{xy}=-2G\).  Note that this is distinct from the guiding-centre metric of the Hall fluid \cite{Haldane2011}.  The two components are locked together by that relation, so \(G\) fixes the local geometric-phase density.  Transporting the quasihole around the boundary of a region \(\Sigma\) accumulates \(2\int_\Sigma G\) in the present convention.  The same coefficient sets how precisely the quasihole can be located, since the quantum Fisher matrix below is \(4G\mathbf1_2\). Although log-derivative covariance formulas are standard in variational Monte Carlo \cite{Sorella2005}, our~\eqref{eq:kernel} extends these local covariances to finite-separation overlaps and phases reconstructed from one experimentally sampled law.

For a finite reference sample the accessible range is set by the fluctuations of the importance weights.  We have \(P_w=\mathcal R_wP_0\) with \(\mathcal R_w=|r_\mathbf{n}(w)|^2/\Ex_0|r_\mathbf{n}(w)|^2\), so that reconstructing at \(w\) is a reweighting problem \cite{FoulkesEtAl2001}, and as \(w\) moves away from \(w_0\) fewer reference shots carry appreciable weight,
\begin{equation}
 \frac{M_{\rm eff}}{M}\longrightarrow
 e^{-\mathsf D_2(P_w\Vert P_0)}.
 \label{eq:horizon}
\end{equation}
The R\'enyi-2 divergence $D_2(P_w\Vert P_0)$ therefore limits how far one reference ensemble reaches, which is a sampling limit and not a boundary of the kernel identity itself.  Where \(\mathsf D_2\) grows large we simply cover the region with several overlapping reference ensembles, at a cost set by distributional overlap rather than by the ambient Hilbert-space dimension.  Locally \(\mathsf D_2(P_{\lambda+\dd\lambda}\Vert P_\lambda)=F^{\rm occ}_{ab}\dd\lambda^a\dd\lambda^b+O(|\dd\lambda|^3)\) for \(\lambda=(x,y)\), where \(F^{\rm occ}\) is the classical Fisher matrix of the occupation probabilities, so the same occupation Fisher geometry controls that reach.

Occupation readout also has a precise local optimality property.  The two real parameters \((x,y)\) span one complex tangent because \(\ket{D_y\psi}=\ii\ket{D_x\psi}\).  Every local measurement problem can therefore be compressed, without changing outcome probabilities or their first derivatives, to \(\operatorname{span}_{\mathbb C}\{\ket\psi,\ket{D_x\psi}\}\).  The Gill--Massar bound is consequently the qubit bound \(\Tr[\mathcal I^{-1}F^{\rm cl}]\le1\), not a bound set by the \(56\)-dimensional many-body host \cite{GillMassar2000,WangChenYuan2026PRL}. Since the pure-state quantum Fisher matrix is \(\mathcal I=4G\mathbf1_2\) \cite{Paris2009} and \(\operatorname{tr}F^{\rm occ}=4G\),
\begin{equation}
 \Tr\!\left[\mathcal I^{-1}F^{\rm occ}\right]=1.
 \label{eq:optimal}
\end{equation}
In other words, occupation snapshots saturate the normalized joint-information trade-off, and this is not generic.  It holds because the occupation scores are exactly \(2\operatorname{Re}\Delta X\) and \(-2\operatorname{Im}\Delta X\), which follows from the same holomorphic factorization that gives \eqref{eq:kernel}, so the same holomorphic structure that makes the kernel reconstructible also makes occupation readout saturate the bound.  Note that the statement is local and multiparameter: it does not say \(F^{\rm occ}=\mathcal I\), since such an equality would give \(\Tr[\mathcal I^{-1}F^{\rm occ}]=2\) and violate the unit bound, nor that occupation readout is Helstrom optimal for two finitely separated states \cite{Helstrom1976}.

The occupation law at one reference position, together with the known quasihole factor, determines the relational geometry of the entire factorized family.  It reconstructs finite-distance overlaps and phases, the exact projective span, and the local metric and curvature without preparing the remaining states.  The construction is model assisted, in that it applies whenever the configuration-wise quasihole factor is known rather than to an arbitrary Coulomb quasihole.  Within that class, a previously recorded ensemble can replace a sequence of state preparations, and the same principle extends to any state family with known configuration-wise parameter dependence.

\textit{Data availability.---} The Mathematica notebook reproducing the benchmark of Fig.~\ref{fig:reconstruction} has been submitted to the Wolfram Notebook Archive \cite{NotebookArchive}.

\begin{acknowledgments}
\textit{Acknowledgments.---}The author is grateful to the DFG for support through a Walter Benjamin Fellowship, project number 515782239.
\end{acknowledgments}

\clearpage
\onecolumngrid

\begin{center}
\textbf{\large Supplemental Material}\\[4pt]
\textbf{Laughlin quasihole geometry from a single snapshot ensemble}\\[2pt]
Kaushlendra Kumar\\[2pt]
\textit{School of Mathematical Sciences, Queen Mary University of London,
Mile End Road, London E1 4NS, United Kingdom}
\end{center}

\setcounter{equation}{0}
\setcounter{section}{0}
\setcounter{theorem}{0}
\renewcommand{\theequation}{S\arabic{equation}}
\renewcommand{\thesection}{S\Roman{section}}

This Supplemental Material derives the single-reference kernel identity and its geometric and statistical consequences in order, followed by illustration of the Nielsen--Cirac--Sierra (NCS) calculation behind Fig.~1 of the Letter. Everything rests on one structural fact: the quasihole position enters each occupation amplitude through a known multiplicative factor. The consequences are algebraic, with the single exception of the finite-sampling analysis, and two of them are quantities normally obtained by quite different means. In contrast, here the Berry curvature is recovered without transporting the state in time, while the quantum metric is obtained without diagonalizing any Hamiltonian.

\section{Single-reference reconstruction}
\label{sm:sec:setup}

We start with a lattice with fixed \(N_s\) sites at distinct points \(z_i\in\mathbb C\).  A configuration \(\mathbf{n}=(n_1,\ldots,n_{N_s})\) records which of them are occupied, with binary values \(n_i\in\{0,1\}\) leading to a fixed total \(\sum_i n_i=N_p\). The corresponding orthonormal states \(\ket{\mathbf{n}}\) span the many-body space \(\mathcal H_{\rm mb}\).  For the eight-site benchmark of Fig.~1 there are \(\binom83=56\) configurations, so that \(\mathcal H_{\rm mb}\simeq\mathbb C^{56}\), and this basis is held fixed from here on. A quasihole of positive integer strength \(p\) sits at a position \(w=x+\ii y\) that the experimenter controls from outside, and each choice of \(w\) produces one many-body state expanded in that same fixed basis.  The collection of states obtained as \(w\) varies is what we call the quasihole family.  The configurations and the family are different things.  The \(56\) configurations label the basis vectors of \(\mathcal H_{\rm mb}\) and never change, whereas each value of \(w\) picks out one particular vector expanded in that basis.  In this construction \(w\) labels states and is not a further particle coordinate.

For the lattice Laughlin states~\cite{NielsenCiracSierra2012,GlasserEtAl2016}, all \(w\)-dependence factorizes:
\begin{align}
 f_\mathbf{n}(w)&:=\prod_i(w-z_i)^{p n_i},\qquad
 \Phi_\mathbf{n}(w)=\Phi_0(\mathbf{n})f_\mathbf{n}(w),
 \label{sm:eq:family}\\
 \ket{\Phi(w)}&:=\sum_\mathbf{n}\Phi_\mathbf{n}(w)\ket{\mathbf{n}},\qquad
 \ket{\psi(w)}:=\frac{\ket{\Phi(w)}}{\sqrt{Z(w)}},\quad
 Z(w):=\langle\Phi(w)|\Phi(w)\rangle .
 \label{sm:eq:states}
\end{align}
The remaining factor \(\Phi_0(\mathbf{n})\) collects every Laughlin, lattice, background-charge, and gauge contribution that does not move with the quasihole, and is therefore the same for all \(w\).  Since the argument uses only the numerical values of the site positions, those sites neither require a special form such as a periodic array, a Bravais lattice nor do they need to have translation symmetry or a momentum-space description.  For a single quasihole, nonzero single-valued factors that do not depend on the configuration \(\mathbf{n}\) amount to a choice of projective representative and are omitted throughout, while possible multivalued several-quasihole monodromy factors \cite{GlasserEtAl2016} lie outside the present scope.

Beyond the finiteness of the configuration set, which for fixed \(N_s\) and \(N_p\) is automatic, the reconstruction needs one piece of model input and three regularity conditions. The model input is that \(f_\mathbf{n}(w)\) can be evaluated for every recorded configuration at every target position, which is ensured by knowing the quasihole family that determines all the factors \((w-z_i)^p\). Next, the reference position must avoid the sites, \(w_0\neq z_i\), because otherwise \(f_\mathbf{n}(w_0)\) vanishes for every configuration that occupies that site and the ratio \(r_\mathbf{n}(w)\) defined below has no meaning there.  Another requirement is that the normalized overlaps should obey \(Z(w)>0\), since a state of vanishing norm determines no ray to normalize.  Lastly, all states are expanded in the one fixed occupation basis, so that the kernel identity below can be properly applied to compare the amplitude of a given configuration at \(w_0\) with the one at \(w\).

\subsection{Kernel theorem}
\label{sm:sec:kernel}

In this work we denote the average of a configuration-dependent quantity \(A\) over a probability law \(P\) on the finite set of configurations by \(\Ex_P[A]=\sum_\mathbf{n}P(\mathbf{n})A(\mathbf{n})\). Let us prepare the quasihole at a chosen site \(w_0\ne z_i\), and repeat the occupation measurement \(M\) times. The outcomes \(\mathbf{n}^{(1)},\ldots,\mathbf{n}^{(M)}\) are independent samples providing the following probability law and average:
\begin{equation}
 P_0(\mathbf{n}):=P_{w_0}(\mathbf{n})=
 \frac{|\Phi_\mathbf{n}(w_0)|^2}{Z(w_0)},\qquad
 \Ex_0[A]:=\sum_\mathbf{n}P_0(\mathbf{n})A(\mathbf{n}).
 \label{sm:eq:refdist}
\end{equation}
We also define the known amplitude ratio as
\begin{equation}
 r_\mathbf{n}(w):=\frac{f_\mathbf{n}(w)}{f_\mathbf{n}(w_0)},\qquad
 \Phi_\mathbf{n}(w)=\Phi_\mathbf{n}(w_0)r_\mathbf{n}(w).
 \label{sm:eq:ratio}
\end{equation}

\begin{theorem}[Single-reference kernel]
For any target positions \(w,v\), the measured reference law determines the unnormalized overlap kernel
\begin{equation}
 \mathcal K(w,v):=\Ex_0\!\left[
 \overline{r_\mathbf{n}(w)}r_\mathbf{n}(v)\right]
 =\frac{\langle\Phi(w)|\Phi(v)\rangle}{Z(w_0)},
 \qquad \mathcal K(w,w)=\frac{Z(w)}{Z(w_0)}.
 \label{sm:eq:kernel-def}
\end{equation}
Moreover, if \(Z(w),Z(v)>0\), the corresponding normalized overlap is determined by
\begin{equation}
 \langle\psi(w)|\psi(v)\rangle=
 \frac{\mathcal K(w,v)}
 {\sqrt{\mathcal K(w,w)\mathcal K(v,v)}}.
 \label{sm:eq:overlap}
\end{equation}
\end{theorem}

\begin{proof}
The proof is by direct substitution and factor adjustment:
\begin{align}
 \mathcal K(w,v)
 &=\sum_\mathbf{n}\frac{|\Phi_\mathbf{n}(w_0)|^2}{Z(w_0)}
 \frac{\overline{f_\mathbf{n}(w)}}{\overline{f_\mathbf{n}(w_0)}}
 \frac{f_\mathbf{n}(v)}{f_\mathbf{n}(w_0)}\nonumber\\
 &=\frac1{Z(w_0)}
 \sum_\mathbf{n}\overline{\Phi_\mathbf{n}(w)}\Phi_\mathbf{n}(v)
 =\frac{\langle\Phi(w)|\Phi(v)\rangle}{Z(w_0)}.
\label{sm:eq:kernel-proof}
\end{align}
Setting \(v=w\) gives the norm ratio, and dividing by the two norms gives \eqref{sm:eq:overlap}.
\end{proof}

It is worth seeing explicitly why the unknown reference phases disappear. Let us write the reference amplitude in polar form, \(\Phi_\mathbf{n}(w_0)=A_\mathbf{n}e^{\ii\theta_\mathbf{n}}\). The experiment then measures \(P_0(\mathbf{n})=A_\mathbf{n}^2/Z(w_0)\) and thus does not know anything about \(\theta_\mathbf{n}\). Now, the amplitudes at two targets are \(\Phi_\mathbf{n}(w)=A_\mathbf{n}e^{\ii\theta_\mathbf{n}}r_\mathbf{n}(w)\) and \(\Phi_\mathbf{n}(v)=A_\mathbf{n}e^{\ii\theta_\mathbf{n}}r_\mathbf{n}(v)\) due to \eqref{sm:eq:ratio}, so each term of the inner product is
\begin{equation}
 \overline{\Phi_\mathbf{n}(w)}\,\Phi_\mathbf{n}(v)
 =A_\mathbf{n}^2\,\overline{r_\mathbf{n}(w)}\,r_\mathbf{n}(v).
 \label{sm:eq:phasecancel}
\end{equation}
We can now write \(r_\mathbf{n}(w)=|r_\mathbf{n}(w)|e^{\ii\varphi_\mathbf{n}(w)}\), such that the surviving factor contains the relative phase \(e^{\ii[\varphi_\mathbf{n}(v)-\varphi_\mathbf{n}(w)]}\), and is the source of the complex information in \(\mathcal K\). The identity assumes that the known ratios \(r_\mathbf{n}(w)\) give the correct \(w\)-dependence of the amplitudes.  It therefore reconstructs the whole family built from those ratios, but it cannot detect an unknown \(w\)-dependent correction to them, because all the data are taken at \(w_0\).

\subsection{Robustness and readout errors}

The restriction to the factorized family in the above kernel identity is stable in a precise sense: a state that departs slightly from~\eqref{sm:eq:family} gives a kernel that departs correspondingly little. To see this, let the target state predicted by the assumed ratios be
\begin{equation}
 \ket{\widetilde\Phi(w)}:=
 \sum_\mathbf{n}\Phi_\mathbf{n}(w_0)r_\mathbf{n}(w)\ket{\mathbf{n}},
 \label{sm:eq:modeled-continuation}
\end{equation}
while the true target state as \(\ket{\Phi^{\rm true}(w)}=\ket{\widetilde\Phi(w)}+\ket{\delta_w}\), so that \(\ket{\delta_w}\) collects whatever part of the true state the assumed ratios fail to capture.  Since \(r_\mathbf{n}(w_0)=1\), the two agree at the reference and \(\delta_{w_0}=0\). Now set \(A_w=\norm{\widetilde\Phi(w)}\) and \(\epsilon_w=\norm{\delta_w}\), and define \(\mathcal K^{\rm true}(w,v)= \langle\Phi^{\rm true}(w)|\Phi^{\rm true}(v)\rangle/Z(w_0)\), such that the kernel inferred from the reference law obeys
\begin{align}
 \left|\mathcal K^{\rm true}(w,v)-\mathcal K(w,v)\right|
 \le \frac{\epsilon_w A_v+A_w\epsilon_v+\epsilon_w\epsilon_v}{Z(w_0)}.
 \label{sm:eq:stability-bound}
\end{align}
This follows by expanding the true overlap and bounding its three error terms with Cauchy--Schwarz.  If the deviation is measured relative to the modelled state, so that \(\epsilon_w\le\eta_w A_w\) and \(\epsilon_v\le\eta_v A_v\), the right-hand side becomes \(A_wA_v(\eta_w+\eta_v+\eta_w\eta_v)/Z(w_0)\), so the kernel error stays small on the scale \(A_wA_v/Z(w_0)\).  It is worth pointing out the case where the true amplitudes carry an unknown factor \(c_\mathbf{n}\) that does not depend on \(w\), so that \(\Phi^{\rm true}_\mathbf{n}(w)=c_\mathbf{n}\Phi_0(\mathbf{n})f_\mathbf{n}(w)\). Then \(c_\mathbf{n}\) appears in numerator and denominator of \(\Phi^{\rm true}_\mathbf{n}(w)/\Phi^{\rm true}_\mathbf{n}(w_0)\) and cancels, leaving \(r_\mathbf{n}(w)\) unchanged, so \(\delta_w=0\). On the other hand, a genuinely \(w\)-dependent departure from the amplitude via an unknown \(\delta r_\mathbf{n}(w)\) that vanishes at \(w_0\) but not elsewhere, contributes to \(\delta_w\).

Readout errors are a different problem, because they act on the recorded configuration rather than on the amplitude. Mathematically, let \(E_{\mathbf{m}\mathbf{n}}\) be the probability that a detector reports \(\mathbf{m}\) when the true configuration is \(\mathbf{n}\), so that each column of \(E\) sums to one, \(\sum_m E_{\mathbf{m}\mathbf{n}}=1\) and the measured law is \(\widetilde P_0=E P_0\). Since the label attached to a shot may now be wrong, the empirical average instead converges to \(\sum_\mathbf{m}\widetilde P_0(\mathbf{m})\overline{r_\mathbf{m}(w)}r_\mathbf{m}(v) \neq \mathcal K(w,v)\).

There are some ways to get around such issues. If \(E\) is characterized independently, as is routine in quantum-gas microscopy, and is invertible on the relevant configurations, then \(P_0=E^{-1}\widetilde P_0\) removes the readout bias in expectation, although the inversion amplifies statistical noise and at finite sample size can return slightly negative entries. Another possible way out is the fact that the normalized quantities of~\eqref{sm:eq:overlap} reconstructed from two distinct reference positions \(w_0\) and \(w_0'\) must agree wherever both ensembles remain statistically accessible. The unnormalized kernels differ by the constant factor \(Z(w_0')/Z(w_0)\), which is why the test is applied after normalization. It probes the assumed ratios and the detector together, so a disagreement shows that something is inconsistent without saying which of the above two issues is at fault.

\section{Finite-separation geometry and empirical kernel}
\label{sm:sec:global}

Before any derivative is taken, the kernel already fixes the geometry at finite separation.  Writing \(\rho_w=\ket{\psi(w)}\!\bra{\psi(w)}\), the pure-state trace distance
\begin{equation}
 T(\rho_w,\rho_v):=\frac12\norm{\rho_w-\rho_v}_1
 =\sqrt{1-|\langle\psi(w)|\psi(v)\rangle|^2}
 \label{sm:eq:tracedist}
\end{equation}
follows directly from \(\mathcal K\)~\eqref{sm:eq:overlap}, as do the gauge-invariant Bargmann products \cite{Bargmann1964}
\begin{equation}
 \mathcal B_{1\cdots L}:=
 \langle\psi_1|\psi_2\rangle\langle\psi_2|\psi_3\rangle\cdots
 \langle\psi_L|\psi_1\rangle ,
 \label{sm:eq:bargmann}
\end{equation}
whose higher-point values determine the projective geometry of the embedded curve \cite{AvdoshkinPopov2023}. Fixing the orientation sign needs a convention, and we take
\begin{equation}
 \mathcal A=\ii\langle\psi|\dd\psi\rangle,\qquad
 \Omega=\dd\mathcal A.
\label{sm:eq:berry-connection}
\end{equation}
For a short directed segment, \(\arg\langle\psi(\lambda)|\psi(\lambda+\dd\lambda)\rangle =-\mathcal A_a\dd\lambda^a+O(|\dd\lambda|^2)\).  A fine polygon approximating an oriented loop \(C\) therefore obeys
\begin{equation}
 \arg\mathcal B_C\longrightarrow-\oint_C\mathcal A
 =-\int_{\Sigma_C}\Omega\pmod{2\pi},
\label{sm:eq:bargmann-berry}
\end{equation}
where \(\Sigma_C\) is any surface whose boundary is \(C\). The choice of the loop $C$ is immaterial modulo \(2\pi\), and reversing the ordering of the vertices reverses the sign. Invariants of this kind are usually obtained through coherent cycle tests performed on several separately prepared states \cite{OszmaniecBrodGalvao2024}, whereas here we reconstruct, via~\eqref{sm:eq:kernel-def}, every factor in the product from one fixed-basis reference distribution.

\subsection{Finite-shot estimator and positivity}

The finite-shot estimator
\begin{equation}
 \widehat{\mathcal K}(w,v)=\frac1M\sum_{s=1}^M
 \overline{r_{\mathbf{n}^{(s)}}(w)}r_{\mathbf{n}^{(s)}}(v)
 \label{sm:eq:kernel-est}
\end{equation}
is unbiased and remains a valid Gram kernel before any large-sample limit is taken.

\begin{proposition}[Finite-shot positivity]
For positions \(w_1,\ldots,w_L\), define \(\boldsymbol u_s=(\overline{r_{\mathbf{n}^{(s)}}(w_1)},\ldots, \overline{r_{\mathbf{n}^{(s)}}(w_L)})^{\!\top}\). Then
\begin{equation}
 \widehat{\mathcal K}_{ab}=\widehat{\mathcal K}(w_a,w_b),
 \qquad \widehat{\mathcal K}=\frac1M\sum_{s=1}^M
 \boldsymbol u_s\boldsymbol u_s^\dagger\succeq0.
 \label{sm:eq:kernel-psd}
\end{equation}
\end{proposition}

\begin{proof}
Each term \(\boldsymbol u_s\boldsymbol u_s^\dagger\) is a rank-one positive matrix, and a positive sum remains positive.
\end{proof}

The conjugation in \eqref{sm:eq:kernel-psd} is the same as in \eqref{sm:eq:kernel-def}, so any finite data set returns a valid Gram matrix.  The estimator \(\widehat{\mathcal K}\) is unbiased, whereas the normalized overlap obtained by inserting it into \eqref{sm:eq:overlap} is a ratio of estimates and therefore carries a small bias at finite \(M\), which vanishes as \(M\to\infty\) whenever the two diagonal entries are nonzero.  Ordinary importance reweighting recovers expectation values at the target \cite{FoulkesEtAl2001}, whereas keeping the two arguments of the kernel independent recovers the off-diagonal overlaps and their phases.

\section{Exact projective compression}
\label{sm:sec:compress}

The quasihole family sits inside \(\mathcal H_{\rm mb}\), whose dimension grows exponentially with \(N_s\), but it does not fill more than a small part of it.  We show here that the states \(\ket{\Phi(w)}\) span at most \(pN_p+1\) dimensions, regardless of \(w\), and that this follows from the same factorization used above. Each amplitude factor is a polynomial of degree \(d=pN_p\), and thus with \(f_\mathbf{n}(w)=\sum_{k=0}^{d}a_k(\mathbf{n})w^k\) and using \eqref{sm:eq:states} we get,
\begin{equation}
 \ket{\Phi(w)}=\sum_{k=0}^{d}w^k\ket{\chi_k},\qquad
 \ket{\chi_k}=\sum_\mathbf{n}\Phi_0(\mathbf{n})a_k(\mathbf{n})\ket{\mathbf{n}}.
 \label{sm:eq:vk}
\end{equation}
The whole family therefore lies in \(\operatorname{span}\{\chi_0,\ldots,\chi_d\}\), and using \(m_d(w)=(1,w,\ldots,w^d)^{\!\top}\) the kernel \eqref{sm:eq:kernel-def} becomes
\begin{equation}
 \mathcal K(w,v)=\frac{m_d(w)^\dagger Hm_d(v)}{Z(w_0)},\qquad
 H_{kl}:=\langle\chi_k|\chi_l\rangle
 =\sum_\mathbf{n}|\Phi_0(\mathbf{n})|^2
 \overline{a_k(\mathbf{n})}a_l(\mathbf{n}).
 \label{sm:eq:Hmatrix}
\end{equation}
\begin{theorem}[Exact projective compression]
\label{sm:thm:compress}
The quasihole family spans a complex subspace of dimension \(\operatorname{rank}H\le d+1=pN_p+1\).  If \(\operatorname{rank}H=r+1\), its projective image lies in a minimal \(\mathbb{CP}^{r}\), with \(r\le pN_p\).  The matrix \(H\) is positive semidefinite.
\end{theorem}

\begin{proof}
For any \(c\in\mathbb C^{d+1}\),
\begin{equation}
 c^\dagger Hc=\norm{\sum_{k=0}^{d}c_k\ket{\chi_k}}^2\ge0,
 \label{sm:eq:Hpositive}
\end{equation}
so \(H\succeq0\), and the right-hand side vanishes only when the combination \(\sum_kc_k\ket{\chi_k}\) itself vanishes.  A vector annihilated by \(H\) is therefore the same thing as a linear relation among the \(\ket{\chi_k}\), so the rank of \(H\) counts how many of those \(d+1\) vectors are linearly independent,
\begin{equation}
 \operatorname{rank}H
 =\dim_{\mathbb C}\operatorname{span}\{\chi_0,\ldots,\chi_d\}\le d+1.
 \label{sm:eq:rankspan}
\end{equation}
The span of the \(\ket{\chi_k}\) coincides with that of the family itself.  Every \(\ket{\Phi(w)}\) is a combination of the \(\ket{\chi_k}\) by \eqref{sm:eq:vk}, and conversely the \(\ket{\chi_k}\) are recovered from finitely many members of the family.  Evaluating \eqref{sm:eq:vk} at \(d+1\) distinct positions \(\xi_0,\ldots,\xi_d\) gives \(\ket{\Phi(\xi_j)}=\sum_kV_{kj}\ket{\chi_k}\), where \(V_{kj}=\xi_j^k\) is a Vandermonde matrix and is invertible whenever the \(\xi_j\) are distinct.  Inverting it writes each \(\ket{\chi_k}\) as a combination of those \(d+1\) states.  The two spans therefore coincide, and passing from vectors to rays lowers the dimension by one.
\end{proof}

Here the minimal projective ambient space means the smallest projective linear subspace containing every ray \([\Phi(w)]\) of the quasihole family.  If \(\operatorname{rank}H=r+1\), the linear span has dimension \(r+1\) and that subspace is \(\mathbb{CP}^{r}\).  Inside it the family traces the curve
\begin{equation}
 w\longmapsto[\,\chi_0+w\chi_1+\cdots+w^d\chi_d\,],
 \label{sm:eq:curve}
\end{equation}
which is holomorphic and of one complex dimension.  That is a different statement from \(r\): the curve stays one-dimensional however large its linear span happens to be, and \(r\) records only how many independent directions the span needs.  Rank one is the degenerate case in which the whole family collapses to a single fixed ray.  The bound holds exactly rather than approximately, since it came from the degree of \(f_\mathbf{n}(w)\) and not from truncating anything.

None of this requires knowing the \(\ket{\chi_k}\) themselves, which is what makes the rank an experimentally accessible number.  Take the same \(d+1\) points \(\xi_j\) and the same \(V\) as in the proof, and collect the kernel between them into \(K^{(\xi)}_{jl}:=\mathcal K(\xi_j,\xi_l)\).  Because the columns of \(V\) are the monomial vectors \(m_d(\xi_j)\), \eqref{sm:eq:Hmatrix} says that these \((d+1)^2\) kernel values are
\begin{equation}
 K^{(\xi)}=V^\dagger\widetilde HV,\qquad
 \widetilde H:=\frac{H}{Z(w_0)},
 \label{sm:eq:Kxi}
\end{equation}
and inverting the Vandermonde matrix recovers the rescaled moment matrix from them,
\begin{equation}
 \widetilde H=(V^\dagger)^{-1}K^{(\xi)}V^{-1}.
\label{sm:eq:Hinterpolation}
\end{equation}
If \(Z(w_0)\) is known independently then \(H=Z(w_0)\widetilde H\), and since that positive factor is common to every entry it changes neither the rank nor the projective geometry.  This is a \((d+1)\times(d+1)\) matrix of kernel values between quasihole positions, and it is not the matrix of state vectors used numerically below, which has the same rank but different entries and a different shape.

For the NCS state of the reference calculation below, \(p=1\) and \(N_p=3\), so \(d=3\) and the theorem bounds the span by four dimensions while the ambient space has \(\dim\mathcal H_{\rm mb}=\binom83=56\).  That bound is saturated.  Collecting the many-body state at forty positions \(w_j\) spaced around the circle \(|w|=1.5\) into the columns of
\begin{equation}
 S=\big[\,\ket{\Phi(w_1)}\ \cdots\ \ket{\Phi(w_{40})}\,\big]
 \in\mathbb C^{56\times40},
 \label{sm:eq:Smatrix}
\end{equation}
the normalized singular values of \(S\) are \(\{1,0.6622864,0.5232412,0.4879987\}\) with the remaining thirty-six at the \(10^{-16}\) numerical floor, so \(\operatorname{rank}S=4\) and therefore \(\operatorname{rank}H=4\).  Four generic positions already reveal the full span, and sampling forty provides an overdetermined numerical check.  Hence \(r=3\) and
\begin{equation}
 \{[\Phi(w)]:w\in\mathbb C\}\subset\mathbb{CP}^3\subset\mathbb{CP}^{55},
 \label{sm:eq:CP3}
\end{equation}
where the left-hand side is the entire continuous curve and not only the forty points sampled from it.

\section{Local geometry and calibrated distance}
\label{sm:sec:local}

So far the kernel has been used at two separate positions \(w\) and \(v\).  Bringing them together gives the local geometry, while keeping them apart yields the finite spectral distance derived at the end of the section.

\subsection{Quantum geometric tensor}

We start with coordinates \(\lambda^a=(x,y)\) with \(w=x+\ii y\), and projector \(\Pi_w=\ket{\psi(w)}\!\bra{\psi(w)}\), such that the projected derivatives and the quantum geometric tensor are
\begin{equation}
 \ket{D_a\psi}:=(1-\Pi_w)\partial_a\ket{\psi(w)},\qquad
 \mathcal Q_{ab}:=\langle D_a\psi|D_b\psi\rangle .
 \label{sm:eq:covtangent}
\end{equation}
Its real and imaginary parts give the Fubini--Study metric and the Berry curvature \cite{ProvostVallee1980,KolodrubetzEtAl2017,Berry1984},
\begin{equation}
 g^{\rm FS}_{ab}=\operatorname{Re}\mathcal Q_{ab},\qquad
 \Omega_{ab}=-2\operatorname{Im}\mathcal Q_{ab},
 \label{sm:eq:gOmega}
\end{equation}
with the sign fixed by the Berry-connection convention of \eqref{sm:eq:berry-connection}.  For the quasihole family both reduce to a single real function of \(w\).

The occupation law at a target \(w\), and the average against it, are
\begin{equation}
 P_w(\mathbf{n})=\frac{|\Phi_\mathbf{n}(w)|^2}{Z(w)},\qquad
 \langle A\rangle_w:=\sum_\mathbf{n}P_w(\mathbf{n})A(\mathbf{n}),
 \label{sm:eq:Pw}
\end{equation}
of which \eqref{sm:eq:refdist} is the case \(w=w_0\). Each configuration carries the logarithmic derivative of its amplitude factor,
\begin{equation}
 X_\mathbf{n}(w):=\partial_w\log f_\mathbf{n}(w)
 =p\sum_i\frac{n_i}{w-z_i},\qquad
 \Delta X:=X_\mathbf{n}(w)-\langle X\rangle_w,
 \label{sm:eq:X}
\end{equation}
evaluated classically from the recorded configuration and the known lattice positions. Although \(X_\mathbf{n}\) has a simple pole at each occupied site, the averages below remain finite. To see this, let us consider the case \(n_j=1\) where the weight \(P_w(\mathbf{n})\) carries a factor \(|w-z_j|^{2p}\), so
\begin{equation}
 P_w(\mathbf{n})\,|X_\mathbf{n}(w)|^2=O\!\left(|w-z_j|^{2p-2}\right),
 \qquad w\to z_j,
 \label{sm:eq:polecancel}
\end{equation}
which is bounded for \(p\ge1\).  For the family of \eqref{sm:eq:family}, whose \(\Phi_0(\mathbf{n})\) is nonzero on every configuration, the normalization at a site is
\begin{equation}
 Z(z_j)=\sum_{\mathbf{n}:\,n_j=0}|\Phi_0(\mathbf{n})|^2
 \prod_{i\ne j}|z_j-z_i|^{2pn_i}>0
 \label{sm:eq:Zsite}
\end{equation}
whenever \(N_p<N_s\), so the normalized state and the geometry built from it are regular on the lattice.

Next, we differentiate the normalization \(Z\) and make use of the holomorphy of \(\Phi_\mathbf{n}(w)\), to obtain
\begin{equation}
 \partial_wZ=\sum_\mathbf{n}|\Phi_\mathbf{n}(w)|^2X_\mathbf{n}(w)=Z\langle X\rangle_w,
 \label{sm:eq:dZ}
\end{equation}
i.e. \(\partial_w\log Z=\langle X\rangle_w\). We differentiate this again with respect to \(\bar w\), which acts on \(|\Phi_\mathbf{n}|^2\) and on \(Z^{-1}\) but not on the holomorphic \(X_\mathbf{n}\), and get
\begin{equation}
 \partial_{\bar w}\partial_w\log Z
 =\langle X\overline X\rangle_w-\langle X\rangle_w\langle\overline X\rangle_w
 =\langle|\Delta X|^2\rangle_w .
 \label{sm:eq:ddlogZ}
\end{equation}
We thus obtain the local metric coefficient as
\begin{equation}
 G(w):=\partial_w\partial_{\bar w}\log Z(w)=\Var_wX
 \label{sm:eq:Gdef}
\end{equation}
from a configuration-space variance without differentiating the normalized state. Using \(\mathcal K(w,w)=Z(w)/Z(w_0)\), equivalently \(G(w)=\partial_w\partial_{\bar w}\log\mathcal K(w,w)\).

In order to evaluate the projected derivatives, we now introduce the diagonal operator
\begin{equation}
 \widehat X(w):=\sum_\mathbf{n}X_\mathbf{n}(w)\ket{\mathbf{n}}\!\bra{\mathbf{n}},
 \qquad \Delta\widehat X:=\widehat X-\langle X\rangle_w\mathbf1 .
 \label{sm:eq:Xop}
\end{equation}

Since \(w=x+\ii y\), the Wirtinger derivatives \(\partial_w=\tfrac12(\partial_x-\ii\partial_y)\) and \(\partial_{\bar w}=\tfrac12(\partial_x+\ii\partial_y)\) invert to \(\partial_x=\partial_w+\partial_{\bar w}\) and \(\partial_y=\ii(\partial_w-\partial_{\bar w})\), and satisfy \(\partial_w\partial_{\bar w}=\tfrac14(\partial_x^2+\partial_y^2)\). Moreover, the holomorphicity of \(\ket{\Phi(w)}\) means \(\partial_{\bar w}\ket\Phi=0\), and both real derivatives collapse onto \(\widehat X\),
\begin{equation}
 \partial_x\ket{\Phi}=\widehat X\ket{\Phi},\qquad
 \partial_y\ket{\Phi}=\ii\widehat X\ket{\Phi}.
 \label{sm:eq:dPhi}
\end{equation}
Differentiating \(\ket\psi=Z^{-1/2}\ket\Phi\) adds a term along \(\ket\psi\), removed by the projector \(1-\Pi_w\), while \(\bra\psi\widehat X\ket\psi=\langle X\rangle_w\) subtracts the mean. The projected derivatives are therefore
\begin{equation}
 \ket{D_x\psi}=\Delta\widehat X\ket{\psi},\qquad
 \ket{D_y\psi}=\ii\Delta\widehat X\ket{\psi},
 \label{sm:eq:tangents}
\end{equation}
differing only by a factor of \(\ii\).  Substituting into \eqref{sm:eq:covtangent} and using \(\bra{\psi}\Delta\widehat X^\dagger\Delta\widehat X\ket{\psi}=\langle|\Delta X|^2\rangle_w=G\) fixes every component,
\begin{equation}
 \mathcal Q=G\begin{pmatrix}1&\ii\\-\ii&1\end{pmatrix},
 \qquad g^{\rm FS}_{ab}=G\delta_{ab},
 \qquad \Omega_{xy}=-2G.
 \label{sm:eq:Q}
\end{equation}

Integrating the curvature over a region \(\Sigma\) and using \eqref{sm:eq:bargmann-berry},
\begin{equation}
 \arg\mathcal B_{\partial\Sigma}=-\int_\Sigma\Omega
 =2\int_\Sigma G\,\dd x\,\dd y,
 \label{sm:eq:phasedensity}
\end{equation}
so \(G\) is also the local geometric-phase density.  For pure states \(\mathcal I_{ab}=4g^{\rm FS}_{ab}\) \cite{BraunsteinCaves1994,Paris2009}, so using \eqref{sm:eq:Q} gives
\begin{equation}
 \mathcal I=4G\,\mathbf1_2,
 \label{sm:eq:QFI}
\end{equation}
implying that the same scalar sets the quantum Fisher scale for local quasihole estimation.  Covariance representations of this kind in logarithmic wave-function derivatives are standard in variational Monte Carlo \cite{Sorella2005}.  The kernel construction extends them to finite-separation overlaps and geometric phases reconstructed from a single reference ensemble.

Here \(g^{\rm FS}_{ab}\) is the metric induced on the quasihole-position parameter manifold, and is not Haldane's guiding-centre metric \cite{Haldane2011}, which characterizes the spatial geometry of correlations in the Hall fluid. On the other hand, if we had snapshots taken directly at \(w\), the resulting coefficient would follow directly from the unbiased sample variance
\begin{equation}
 \widehat G_{\rm direct}(w)=\frac1{M-1}\sum_{s=1}^M
 \big|X_{\mathbf{n}^{(s)}}(w)-\overline X\big|^2,
 \qquad \overline X:=\frac1M\sum_{s=1}^M X_{\mathbf{n}^{(s)}}(w).
 \label{sm:eq:Gdirect}
\end{equation}
In the single-reference protocol, however, every shot is drawn at \(w_0\) and follows \(P_0\) rather than \(P_w\), so \eqref{sm:eq:Gdirect} is not available at any target other than the reference itself.  Recovering \(G(w)\) from those shots requires reweighting them towards \(P_w\). We construct precisely this in the next section and this also fixes how far from \(w_0\) the reconstruction remains usable.

\subsection{Scalar-anchored distance}
\label{sm:sec:anchor}

We have already seen \eqref{sm:eq:tracedist} how the kernel fixes the trace distance. The scalar-anchor spectral triple equips that quantity with a Connes-distance interpretation, at the cost of one externally calibrated length scale. This was already demonstrated for the qubit matrix sector in \cite[Appendix C]{Kumar2026anchor}, and the present setting requires only enlarging the matrix summand from \(M_2(\mathbb C)\) to \(M_{N_{\rm H}}(\mathbb C)\). The latter acts on the fixed occupation basis \(\{\ket{\mathbf{n}}\}\) and has size
\begin{equation}
 N_{\rm H}:=\dim_{\mathbb C}\mathcal H_{\rm mb}=\binom{N_s}{N_p},
 \qquad \mathcal H_{\rm mb}\simeq\mathbb C^{N_{\rm H}},
 \label{sm:eq:NH}
\end{equation}
which is \(56\) for the eight-site benchmark.  A finite spectral triple is the data \((A,\mathcal H,D)\) consisting of a \(*\)-algebra \(A\) represented by \(\pi\) on a finite-dimensional Hilbert space \(\mathcal H\), together with a self-adjoint Dirac operator \(D\) acting on \(\mathcal H\). The commutator \([D,\pi(a)]\) measures how fast an algebra element varies across the space and provides the Lipschitz seminorm from which the distance is built. The data used here are
\begin{align}
 A&=M_{N_{\rm H}}(\mathbb C)\oplus\mathbb C,
 &\mathcal H&=\mathbb C^{N_{\rm H}}\oplus\mathbb C^{N_{\rm H}},
 \nonumber\\
 \pi(a,\beta)&=\begin{pmatrix}a&0\\0&\beta I\end{pmatrix},
 &D_\Lambda&=\Lambda\begin{pmatrix}0&I\\I&0\end{pmatrix},\qquad\Lambda>0,
 \label{sm:eq:spectraldata}
\end{align}
with \((a,\beta)\in A\) and \(I\) the identity on \(\mathbb C^{N_{\rm H}}\).  The matrix summand carries the physical states and the scalar summand \(\beta\in\mathbb C\) is the anchor, represented as \(\beta I\) on the second copy of \(\mathbb C^{N_{\rm H}}\). Here \(D_\Lambda\) is the Dirac operator defining the Lipschitz seminorm in the Connes distance \cite{Connes1994},
\begin{equation}
 d_\Lambda(\rho,\sigma)
 =\sup\left\{\big|\Tr[(\rho-\sigma)a]\big|\;:\;
 a=a^\dagger,\ \beta\in\mathbb R,\
 \norm{[D_\Lambda,\pi(a,\beta)]}\le1\right\}.
 \label{sm:eq:connes-sup}
\end{equation}
Since
\begin{equation}
 [D_\Lambda,\pi(a,\beta)]=\Lambda\begin{pmatrix}
 0&\beta I-a\\ a-\beta I&0\end{pmatrix},\qquad
 \norm{[D_\Lambda,\pi(a,\beta)]}=\Lambda\norm{a-\beta I}_{\rm op},
 \label{sm:eq:lipschitz}
\end{equation}
the admissible elements are those with \(\norm{a-\beta I}_{\rm op}\le1/\Lambda\).  The objective only involves \(\Delta:=\rho-\sigma\), which is traceless, so \(\beta\) cancels and the supremum may be taken over \(b:=\Lambda(a-\beta I)\) with \(\norm{b}_{\rm op}\le1\).  For Hermitian \(\Delta\) the two norms are dual, \(\sup_{\norm{b}_{\rm op}\le1}|\Tr(b\Delta)|=\norm{\Delta}_1\), where the supremum is attained at \(b=\sum_i\operatorname{sgn}(\mu_i)\ket{i}\!\bra{i}\) in the eigenbasis \(\Delta=\sum_i\mu_i\ket{i}\!\bra{i}\). Following the argument of Ref.~\cite[Appendix C]{Kumar2026anchor} we thus obtain
\begin{equation}
 d_\Lambda(\rho,\sigma)=\frac{\norm{\rho-\sigma}_1}{\Lambda}
 =\frac{2}{\Lambda}T(\rho,\sigma),
 \label{sm:eq:connes}
\end{equation}
with \(T\) the trace distance.  For pure states,
\begin{equation}
 d_\Lambda(\rho_\psi,\rho_\phi)=\frac{2}{\Lambda}
 \sqrt{1-|\langle\psi|\phi\rangle|^2},
 \label{sm:eq:purechord}
\end{equation}
and the reconstructed overlap \eqref{sm:eq:overlap} gives the calibrated distance between two quasihole positions,
\begin{equation}
 d_\Lambda(\rho_w,\rho_v)=\frac{2}{\Lambda}
 \sqrt{1-\frac{|\mathcal K(w,v)|^2}
 {\mathcal K(w,w)\mathcal K(v,v)}}.
 \label{sm:eq:kernelchord}
\end{equation}

The scalar summand supplies one state that the matrix sector does not contain, the character
\begin{equation}
 \rho_\bullet(a,\beta)=\beta ,
 \label{sm:eq:anchor}
\end{equation}
which evaluates the scalar and discards the matrix block.  The same Connes metric can be evaluated between a physical density-matrix state and this character, a pair of a different kind from the two matrix-sector states of \eqref{sm:eq:connes}, though both are evaluations of the one spectral distance.  Unit trace gives \(\Tr(\rho a)-\beta=\Tr[\rho(a-\beta I)]\) for any density matrix \(\rho\), so the same substitution \(b=\Lambda(a-\beta I)\) yields
\begin{equation}
 d_\Lambda(\rho,\rho_\bullet)
 =\Lambda^{-1}\sup_{\norm{b}_{\rm op}\le1}|\Tr(\rho b)|
 =\frac{\norm{\rho}_1}{\Lambda}=\frac1\Lambda ,
 \label{sm:eq:anchordist}
\end{equation}
since \(\rho\ge0\) and \(\Tr\rho=1\), with the supremum attained at \(b=I\).  Equation \eqref{sm:eq:connes} therefore gives distances within the physical matrix sector, whereas \eqref{sm:eq:anchordist} gives the distance from that sector to the scalar anchor.  The scalar state is a common metric anchor lying at \(1/\Lambda\) from every state of the matrix sector, while separations inside that sector stay below \(2/\Lambda\).  Because \(\rho_\bullet\) belongs to the scalar sector and not to \(\mathcal H_{\rm mb}\), no measurement on the quasihole family can fix that separation, which is why \(\Lambda\) has to be calibrated externally.

It should be noted here that the spectral triple does not introduce a new measure of distinguishability, since \eqref{sm:eq:connes} returns the trace distance already reconstructed from the kernel. Instead, it gives a spectral-geometric origin for that quantity together with a single calibration constant \(\Lambda\), the off-diagonal strength in \(D_\Lambda\), fixed by one external reference distance rather than inferred from the snapshot ensemble. It sets the unit of length only, since for two physical states \eqref{sm:eq:connes} gives \(T(\rho,\sigma)=\tfrac{\Lambda}{2}d_\Lambda(\rho,\sigma)\) and hence
\begin{equation}
 P_{\rm succ}^{\rm opt}=\tfrac12+\tfrac12T(\rho,\sigma)
 =\tfrac12+\tfrac{\Lambda}{4}d_\Lambda(\rho,\sigma)
 \label{sm:eq:helstrom}
\end{equation}
for equal-prior discrimination \cite{Helstrom1976}, in which \(\Lambda\) cancels, so Helstrom distinguishability is calibration independent.

For nearby positions, expanding \eqref{sm:eq:purechord} with \(1-|\langle\psi(\lambda)|\psi(\lambda+\dd\lambda)\rangle|^2 =g^{\rm FS}_{ab}\dd\lambda^a\dd\lambda^b+O(|\dd\lambda|^3)\) gives the line element
\begin{equation}
 \dd s_\Lambda^2=\frac{4}{\Lambda^2}g^{\rm FS}_{ab}
 \dd\lambda^a\dd\lambda^b
 =\frac{1}{\Lambda^2}\mathcal I_{ab}\dd\lambda^a\dd\lambda^b
 =\frac{4G(w)}{\Lambda^2}|\dd w|^2 .
 \label{sm:eq:line}
\end{equation}
Equation~\eqref{sm:eq:kernelchord} is the finite endpoint distance determined by the trace norm, whereas integrating \(\dd s_\Lambda\) gives the length of a path along the quasihole manifold.  The two coincide only infinitesimally.

\section{Statistical reach and local information}
\label{sm:sec:horizon}

The preceding identities assume the exact reference law \(P_0\).  With \(M\) reference snapshots, reconstruction away from \(w_0\) becomes an importance-sampling problem, and the quantity that limits its range also controls the local information carried by a single snapshot.

\subsection{Finite-sample reconstruction}

Equations \eqref{sm:eq:refdist} and \eqref{sm:eq:ratio} give the ratio of the occupation law at a target \(w\) to the law that is actually sampled,
\begin{equation}
 \mathcal R_w(\mathbf{n}):=\frac{P_w(\mathbf{n})}{P_0(\mathbf{n})}
 =\frac{|r_\mathbf{n}(w)|^2}{\Ex_0|r_\mathbf{n}(w)|^2},
 \qquad \Ex_0[\mathcal R_w]=1,
 \label{sm:eq:ratioR}
\end{equation}
so that \(\Ex_w[A]=\Ex_0[\mathcal R_wA]\) for any configuration-dependent \(A\).  The complex \(r_\mathbf{n}(w)\) of \eqref{sm:eq:ratio} reconstructs off-diagonal overlaps and carries the relative phases, whereas its normalized modulus squared is the positive \(\mathcal R_w\) that reweights diagonal expectations only. A finite sample replaces this exact reweighting by empirical weights. Writing \(\kappa_s:=|r_{\mathbf{n}^{(s)}}(w)|^2\) for the unnormalized weight of the \(s\)th shot,
\begin{equation}
 \alpha_s:=\frac{\kappa_s}{\sum_t\kappa_t},\qquad
 M_{\rm eff}:=\frac1{\sum_s\alpha_s^2}
 =\frac{(\sum_s\kappa_s)^2}{\sum_s\kappa_s^2},
 \label{sm:eq:finiteESS}
\end{equation}
and since \(M^{-1}\sum_t\kappa_t\to\Ex_0[|r_\mathbf{n}(w)|^2]\), the product \(M\alpha_s\) converges to \(\mathcal R_w(\mathbf{n}^{(s)})\).  The effective sample size interpolates between the extremes of weight concentration, returning \(M_{\rm eff}=M\) when all weights are equal and \(M_{\rm eff}\to1\) when one shot carries them all.

What controls that concentration is the R\'enyi-2 divergence \cite{Renyi1961}.  Substituting \(P_w=P_0\mathcal R_w\) in its definition gives
\begin{equation}
 \mathsf D_2(P_w\Vert P_0):=
 \log\sum_\mathbf{n}\frac{P_w(\mathbf{n})^2}{P_0(\mathbf{n})}
 =\log\Ex_0[\mathcal R_w^2],
 \label{sm:eq:renyi-def}
\end{equation}
which vanishes at \(P_w=P_0\) and grows as the weights spread.  Dividing numerator and denominator in \eqref{sm:eq:finiteESS} by \(M\) and using \(M^{-1}\sum_s\kappa_s\to\Ex_0[|r_\mathbf{n}(w)|^2]\) together with \(M^{-1}\sum_s\kappa_s^2\to\Ex_0[|r_\mathbf{n}(w)|^4]\),
\begin{equation}
 \frac{M_{\rm eff}}{M}\longrightarrow
 \frac{\big(\Ex_0[|r_\mathbf{n}(w)|^2]\big)^2}{\Ex_0[|r_\mathbf{n}(w)|^4]}
 =\frac{1}{\Ex_0[\mathcal R_w^2]}
 =e^{-\mathsf D_2(P_w\Vert P_0)} .
 \label{sm:eq:Meff}
\end{equation}
At the far end of the ray in Fig.~1(b) the divergence reaches \(0.6929\simeq\log2\), so \(M_{\rm eff}\simeq M/2\) and the \(4\times10^3\) recorded shots carry the statistical weight of about half that many equally weighted ones. This is not discarding half of the samples, rather it is a reduction of efficiency. For a given target and reference the finite-shot cost is therefore set by \(\mathsf D_2\) and not by the ambient Hilbert-space dimension. Nevertheless \(\mathsf D_2\) itself can grow with system size or displacement, in which case a region of interest is covered by several overlapping reference ensembles.

Since \(\Ex_w[A]\) is approximated by \(\sum_s\alpha_sA_{\mathbf{n}^{(s)}}\), the target variance \(G=\Var_wX\) is estimated by the same expression evaluated on the weighted sample,
\begin{equation}
 \widehat G_{\rm rw}(w)=\sum_s\alpha_s|X_{\mathbf{n}^{(s)}}(w)|^2
 -\left|\sum_s\alpha_sX_{\mathbf{n}^{(s)}}(w)\right|^2,
 \label{sm:eq:reweightG}
\end{equation}
which uses only snapshots drawn at \(w_0\) and is nonnegative because it is a weighted mean square taken about a weighted mean.  The weights carry the random denominator \(\sum_t\kappa_t\), so \eqref{sm:eq:reweightG} is a ratio of sample averages, consistent as \(M\to\infty\) but not exactly unbiased at finite \(M\) \cite{Liu2001}.  Figure~1(b) shows one finite-shot reconstruction, with error bars giving the standard deviation over \(160\) independent reference ensembles.

\subsection{Local information trade-off}
\label{sm:sec:fisher}

The same divergence evaluated between two nearby targets (unlike between a target and the reference as above) produces the occupation Fisher matrix as its quadratic form. Recall the real quasihole coordinates \(\lambda^a=(x,y)\), so that \(P_\lambda\) is the occupation law \(P_w\) of \eqref{sm:eq:Pw} read as a function of those coordinates, and define
\begin{equation}
 s_a(\mathbf{n};\lambda):=\partial_a\log P_\lambda(\mathbf{n}),\qquad
 F^{\rm occ}_{ab}:=\Ex_\lambda[s_as_b].
 \label{sm:eq:score}
\end{equation}
Now expanding \(P_{\lambda+\dd\lambda}/P_\lambda=1+s_a\dd\lambda^a+\tfrac12t_{ab}\dd\lambda^a\dd\lambda^b+O(|\dd\lambda|^3)\) with \(t_{ab}:=(\partial_a\partial_bP_\lambda)/P_\lambda\) followed by squaring and averaging while noting that \(\Ex_\lambda[s_a]=0\) and \(\Ex_\lambda[t_{ab}]=0\) due to normalisation we find the score covariance as the first surviving term,
\begin{equation}
 \mathsf D_2(P_{\lambda+\dd\lambda}\Vert P_\lambda)
 =F^{\rm occ}_{ab}\dd\lambda^a\dd\lambda^b+O(|\dd\lambda|^3).
 \label{sm:eq:renyi-local}
\end{equation}

Furthermore, we have \(\log P_w=\log|\Phi_\mathbf{n}(w)|^2-\log Z(w)\) as well as holomorphicity of \(f_\mathbf{n}(w)\), so differentiating \eqref{sm:eq:Pw} returns the centred variable \(\Delta X\) of \eqref{sm:eq:X},
\begin{equation}
 s_x=2\operatorname{Re}\Delta X,\qquad
 s_y=-2\operatorname{Im}\Delta X,\qquad
 B(w):=\langle(\Delta X)^2\rangle_w .
 \label{sm:eq:scoresXY}
\end{equation}
Here \(B\) supplies the full Fisher matrix together with their correlation, since the scalar \(G=\langle|\Delta X|^2\rangle_w\) fixes only the sum of the two second moments of \(\Delta X\). We can now write \(\Delta X=a+\ii b\) such that \(\operatorname{Re}B=\langle a^2-b^2\rangle_w\) and \(\operatorname{Im}B=2\langle ab\rangle_w\), which leads to \(\langle a^2\rangle_w=(G+\operatorname{Re}B)/2\) and \(\langle b^2\rangle_w=(G-\operatorname{Re}B)/2\). In other words, averaging the outer products of \eqref{sm:eq:scoresXY} gives
\begin{equation}
 F^{\rm occ}=2\begin{pmatrix}
 G+\operatorname{Re}B&-\operatorname{Im}B\\
 -\operatorname{Im}B&G-\operatorname{Re}B
 \end{pmatrix}.
 \label{sm:eq:Focc}
\end{equation}
It is worth noting that even though the full Fisher matrix is anisotropic, its trace is not: \(\operatorname{tr}F^{\rm occ}=4G\). Also, its eigenvalues \(2(G\pm|B|)\) are both nonnegative because \(|B|\le G\) by Cauchy--Schwarz. Thus \(B\) measures how unevenly the information is spread over the \((x,y)\) plane, ranging from \(F^{\rm occ}=2G\mathbf1_2\) at \(B=0\) to a single informative direction with eigenvalues \(4G\) and \(0\) at \(|B|=G\).

We now ask how large that trace can be for an arbitrary measurement, which requires comparing it against the quantum Fisher matrix \eqref{sm:eq:QFI}. For a POVM \(\{M_\mu\}\) with outcome probabilities \(p_\mu(w)=\bra{\psi(w)}M_\mu\ket{\psi(w)}\), the classical Fisher matrix reads
\begin{equation}
 F^{\rm cl}_{ab}[M]=\sum_\mu p_\mu
 (\partial_a\log p_\mu)(\partial_b\log p_\mu).
 \label{sm:eq:Fgeneral}
\end{equation}
Note that the measurements optimal for \(x\) and for \(y\) separately are in general incompatible, so we cannot expect \(F^{\rm cl}=\mathcal I\) and must look instead for a trade-off between the two directions. To that end, let us recall from \eqref{sm:eq:tangents} that \(\ket{D_y\psi}=\ii\ket{D_x\psi}\), so the state together with its entire first-order tangent space lies in
\begin{equation}
 \mathcal S_w:=\operatorname{span}_{\mathbb C}
 \{\ket\psi,\ket{D_x\psi}\},
 \label{sm:eq:local-subspace}
\end{equation}
which has complex dimension two whenever \(G>0\), since \(\ket{D_x\psi}\) is orthogonal to \(\ket\psi\) and has norm \(\sqrt G\). Locally the estimation problem is therefore a qubit problem sitting inside the far larger \(\mathcal H_{\rm mb}\). To see that nothing is lost in this reduction, let \(P_{\mathcal S}\) project onto \(\mathcal S_w\) and compress an arbitrary POVM to \(\widetilde M_\mu=P_{\mathcal S}M_\mu P_{\mathcal S}\). The compressed operators sum to \(P_{\mathcal S}\) and hence form a POVM on \(\mathcal S_w\), and since \(P_{\mathcal S}\ket\psi=\ket\psi\) and \(P_{\mathcal S}\ket{D_a\psi}=\ket{D_a\psi}\), they leave both \(p_\mu\) and \(\partial_ap_\mu\) untouched. Every Fisher contribution from an outcome with \(p_\mu>0\) thus survives the compression, and an outcome with \(p_\mu=0\) forces \(M_\mu\ket\psi=0\) so that its leading variation is quadratic, which the compression preserves as well. The Gill--Massar inequality for a two-dimensional Hilbert space now applies to any single-copy POVM \cite{GillMassar2000},
\begin{equation}
 \Tr\!\left(\mathcal I^{-1}F^{\rm cl}[M]\right)\le1,
 \label{sm:eq:GM}
\end{equation}
which with \(\mathcal I^{-1}=(4G)^{-1}\mathbf1_2\) simply reads \(\operatorname{tr}F^{\rm cl}\le4G\). For separable measurements on repeated copies both the classical and the quantum Fisher matrices are additive, so the same normalised unit bound survives. Also, the tight multiparameter trade-off of Ref.~\cite{WangChenYuan2026PRL} reduces to this bound for the present two-parameter pure-state model.

Now for the occupation POVM \(M_\mathbf{n}=\ket{\mathbf{n}}\!\bra{\mathbf{n}}\), the general expression \eqref{sm:eq:Fgeneral} reduces to \(F^{\rm occ}\), so that \(\operatorname{tr}F^{\rm occ}=4G\) gives
\begin{equation}
 \Tr\!\left(\mathcal I^{-1}F^{\rm occ}\right)
 =\frac{\operatorname{tr}F^{\rm occ}}{4G}=1 .
 \label{sm:eq:saturation}
\end{equation}
In other words, occupation readout saturates the single-copy Gill--Massar trade-off for joint local estimation of \((x,y)\). This ultimately follows from the holomorphic quasihole dependence, which produces the score relations \eqref{sm:eq:scoresXY} and hence \(\operatorname{tr}F^{\rm occ}=4G\), precisely the ceiling set by \eqref{sm:eq:GM}. It is worth noting that the saturation constrains the trace and not the matrix itself: as \(\operatorname{tr}\mathcal I=8G\), a measurement achieving \(F^{\rm cl}=\mathcal I\) would give \(\Tr(\mathcal I^{-1}\mathcal I)=2\) and violate \eqref{sm:eq:GM}, so no single-copy measurement reaches the quantum Fisher matrix in both directions at once, and indeed \(F^{\rm occ}\ne\mathcal I\) for every \(G>0\). Nevertheless, occupation readout does extract the largest total that any such measurement can. Finally, all of this concerns local estimation rather than finite-separation Helstrom discrimination. Occupation snapshots do reconstruct the Helstrom distance between two separated positions through the kernel, but they are not in general the optimal measurement for discriminating those two states.

\section{Reference calculation for Fig.~1}
\label{sm:sec:numerics}

\textit{State and parameters.---} We enumerate all \(56\) occupation configurations exactly, so that the only random operation of concern is multinomial sampling from the reference law, without ever bothering about the Hamiltonian or its diagonalisation. The lattice is the generation-one Sierpi\'nski carpet, namely the eight points \(z_i=x_i+\ii y_i\) with \(x_i,y_i\in\{-1,0,1\}\) and \((x_i,y_i)\ne(0,0)\), and we take the NCS Laughlin parameter \(q=2\), \(N_p=3\) hard-core bosons and one quasihole of strength \(p=1\) \cite{NielsenCiracSierra2012,GlasserEtAl2016}, for which NCS neutrality \(qN_p+p=\eta N_s\) fixes the background charge at \(\eta=7/8\). Up to fixed gauge phases, which cancel from every quantity used here,
\begin{equation}
 \Phi_\mathbf{n}(w)=\delta_{\sum_i n_i,N_p}
 \prod_i(w-z_i)^{p n_i}
 \prod_{i<j}(z_i-z_j)^{q n_i n_j}
 \prod_{k\ne l}(z_k-z_l)^{-\eta n_l},
 \label{sm:eq:ncs}
\end{equation}
with the principal branch taken for the fixed background powers, whose phases likewise cancel from every reconstructed observable. Recall that this gives \(\dim\mathcal H_{\rm mb}=\binom83=56\).

\textit{Metric and finite-sample reconstruction.---} At \(w=0.35-0.20\ii\) we evaluate the same analytic identity in three independent ways, through the covariance \eqref{sm:eq:Q}, the normalization derivative \eqref{sm:eq:Gdef}, and \(-\Omega_{xy}/2\) obtained from centered differences of the normalized state:
\begin{align}
 G_{\rm covariance}&=0.7043156995,\nonumber\\
 G_{\partial\bar\partial\log Z}&=0.7043157000,\nonumber\\
 -\Omega_{xy}/2&=0.7043156936.
\label{sm:eq:metric-check}
\end{align}
The first value is exact, while the latter two are independent finite-difference evaluations of \(\partial_w\partial_{\bar w}\log Z\) and \(-\Omega_{xy}/2\), so that their agreement to better than \(6\times10^{-9}\) checks the covariance formula and the sign of the curvature. The background of Fig.~1(a) is \(G\) itself, evaluated from \eqref{sm:eq:Q} on a square mesh of spacing \(0.05\) offset by half a step so that no sample falls on a site, the mesh being a plotting device that leaves the exact covariance definition untouched. Since \(pN_p=3\), the normalization \(Z\) has bidegree at most \((3,3)\) in \((w,\bar w)\), and a monomial \(w^a\bar w^b\) contributes angular dependence unless \(a=b\). The quarter-turn symmetry of the lattice and its NCS weights requires \(a-b\equiv0\pmod4\), which for \(0\le a,b\le3\) leaves only \(a=b\), so that \(Z(w)=a_0+a_1|w|^2+a_2|w|^4+a_3|w|^6\) and \(G\) is exactly radial for this benchmark. Writing \(s=|w|^2\) and using \(\partial_ws=\bar w\) gives \(G=\partial_s\log Z+s\,\partial_s^2\log Z\), whose value at the origin is \(a_1/a_0\) and whose slope there is \(\partial_sG|_{s=0}=4a_2/a_0-2(a_1/a_0)^2>0\). At large \(|w|\) the leading term \(pN_p/w\) of \(X_\mathbf{n}\) is the same for every configuration at fixed \(N_p\) and drops out of the variance, so the first surviving contribution is \(O(|w|^{-2})\) and \(G=O(|w|^{-4})\). Since \(G\) therefore increases away from the origin but decays at large distance, it develops a finite-radius maximum, here at \(|w|\simeq0.907\) with \(G\simeq1.302\).

The ray chosen for Fig.~1(b), \(w(t)=w_0+te^{0.7\ii}\) with \(w_0=0.6-0.3\ii\) and \(t\in[0,2.4]\), crosses this annular maximum near \(t\simeq0.40\), where \(|w|\simeq0.907\) and \(G\simeq1.302\), so Fig.~1(b) is a one-dimensional cut through the radial profile displayed in Fig.~1(a). Taking \(M=4\times10^3\) along it, the plotted points are \eqref{sm:eq:reweightG} evaluated on one such ensemble, and the error bars are the standard deviation over \(160\) independently drawn ensembles. Note that these repetitions only characterize the finite-shot scatter and bias, and are not part of the protocol itself, for which the single prepared reference suffices. We find the bias of the \(160\)-ensemble mean to remain below \(0.5\%\).  Along the ray \(\mathsf D_2(P_w\Vert P_0)\) increases, reducing the effective sample fraction from \(M_{\rm eff}/M=0.748\) at the annular crossing, where \(\mathsf D_2=0.2907\), to \(M_{\rm eff}/M\simeq0.50\) at \(t=2.4\), where \(\mathsf D_2=0.6929\).

\textit{Phase test and projective rank.---} The red triangle in Fig.~1(a) marks the three target positions used for the finite-separation Bargmann-phase test. From each reference ensemble we reconstruct the three overlaps joining its vertices,
\begin{equation}
 \zeta_j=0.55\exp\!\left[\ii\left(0.25+\tfrac{2\pi(j-1)}{3}\right)\right],
 \qquad j=1,2,3,
 \label{sm:eq:triangle}
\end{equation}
all of which are distinct from the prepared reference \(w_0\). The exact phase of \(\mathcal B_{123}\) is \(0.342370\) rad against the reconstructed \(0.3424\pm0.0104\) rad, quoted as the mean and standard deviation over the \(160\) ensembles, which verifies that a gauge-invariant complex phase is recovered from the reference distribution alone. Finally, Fig.~1(c) tests the exact projective compression. We evaluate \(\ket{\Phi(w)}\) at the \(40\) equally spaced positions on \(|w|=1.5\) that form the columns of \(S\) in \eqref{sm:eq:Smatrix}, whose normalized singular values are \(\{1,0.6622864,0.5232412,0.4879987\}\) with the remaining \(36\) lying at the \(10^{-16}\) numerical floor, so that \(\operatorname{rank}S=4=d+1\). Sampling \(40\) positions rather than four provides an overdetermined check of the analytic rank bound of Theorem~\ref{sm:thm:compress}, independent of the overall normalization of \(\ket{\Phi}\). It is worth noting that \(S\) is a different object from the kernel matrix \(K^{(\xi)}\) of \eqref{sm:eq:Kxi}.

The accompanying Mathematica notebook \cite{NotebookArchive} reproduces the benchmark and all numerical results quoted above.

\end{document}